\documentclass[10pt,journal,compsoc]{IEEEtran}
\ifCLASSOPTIONcompsoc
  \usepackage[nocompress]{cite}
\else
  \usepackage{cite}
\fi
\ifCLASSINFOpdf
\else
\fi
\usepackage{amsmath,amssymb}
\usepackage{changepage}
\usepackage[flushleft]{threeparttable}
\usepackage{array,booktabs,makecell}
\usepackage{setspace}
\DeclareMathOperator*{\argmax}{argmax} 
\usepackage{graphicx}
\usepackage{physics}
\newcolumntype{C}[1]{>{\centering\arraybackslash}p{#1}}
\newcolumntype{L}[1]{>{\raggedright\arraybackslash}p{#1}}

\newcommand{\arm}{{U^{(t)}_{\mathrm{ARM}}}_{i}}
\usepackage{graphicx}
\usepackage{xcolor}

\usepackage{array}
\usepackage{multirow}
\usepackage{amsfonts}
\usepackage{enumerate}
\usepackage{amsthm}
\usepackage{svg}
\begin{document}
%
\title{Controlling Segregation in Social Network Dynamics as an Edge Formation Game}
%
%
%
%

\author{Rui~Luo,
        Buddhika~Nettasinghe,
        and~Vikram~Krishnamurthy,~\IEEEmembership{Fellow,~IEEE}
\IEEEcompsocitemizethanks{\IEEEcompsocthanksitem R. Luo is with the Sibley School of Mechanical and Aerospace Engineering, Cornell University, Ithaca, NY, 14850.\protect\\
E-mail: rl828@cornell.edu
\IEEEcompsocthanksitem B. Nettasinghe and V. Krishnamurthy are with the School of Electrical and Computer Engineering, Cornell University, Ithaca, NY, 14850.\protect\\
E-mail: \{dwn26, vikramk\}@cornell.edu
\IEEEcompsocthanksitem This research was supported in part by the U. S. Army Research Office under grants W911NF-21-1-0093 and W911NF-19-1-0365.\protect\\}}

\IEEEtitleabstractindextext{
\begin{abstract}
This paper studies controlling \emph{segregation} in social networks via exogenous incentives. We construct an edge formation game on a directed graph. A user (node) chooses the probability with which it forms an inter- or intra- community edge based on a utility function that reflects the tradeoff between homophily~(preference to connect with individuals that belong to the same group) and the preference to obtain an exogenous incentive. Decisions made by the users to connect with each other determine the evolution of the social network. We explore an \emph{algorithmic recommendation mechanism} where the exogenous incentive in the utility function is based on \emph{weak ties} which incentivizes users to connect across communities and mitigates the segregation. This setting leads to a submodular game with a unique Nash equilibrium. In numerical simulations, we explore how the proposed model can be useful in controlling segregation and echo chambers in social networks under various settings.
\end{abstract}

}

\maketitle

\IEEEdisplaynontitleabstractindextext

%
\IEEEpeerreviewmaketitle

\IEEEraisesectionheading{\section{Introduction}\label{sec:introduction}}

%
%
%
%


\IEEEPARstart{S}ocial networks provide a platform for people to exchange information and form online communities. Subject to the homophily effect \cite{mcpherson2001birds} and the bias in the information towards like-minded peers \cite{musco2018minimizing}, social network users tend to connect with others with similar opinion or partisanship, often leading to segregation. Impacts of segregation in social networks include limiting the exposure to diverse perspectives\cite{cinelli2021echo}, amplifying economic inequality\cite{toth2021inequality}, and reinforcing user's purchase interests in e-commerce\cite{ge2020understanding}.

A formal definition of segregation is as follows. Consider a social network represented by a directed graph $G=(V,E)$ without loops or multiple edges, where $V=R\cup B$ is a partition of the set of users into red (R) community and blue (B) community. Let $E_d$ be the set of edges connecting nodes in different communities. We define the segregation measure as:
\begin{equation} \label{eq:segregation index}
    s = 1 - \frac{|E_d|}{2|R| |B|}
\end{equation}
This segregation measure compares the actual number and the maximal possible number of inter-community edges. When the network is completely segregated, $s=1$. Similar definition of segregation can be found in the \emph{segregation index} defined in \cite{sasahara2019inevitability} and the \emph{assortativity coefficient} defined in \cite{rodriguez2016overview}. 
A vivid example of segregation in a directed graph is from Twitter users' retweeting behavior before the 2020 presidential election\cite{luo2021echo}, as shown in Fig. \ref{fig:twitter polarized}. 

\begin{figure}
	\centering
	\includegraphics[width=0.38\textwidth]{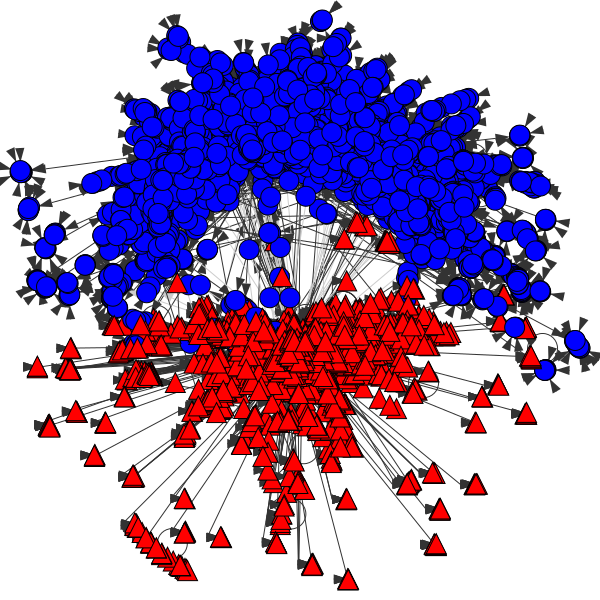}
	\caption{The figure shows that Twitter users' political opinions are segregated into two highly polarized communities before the 2020 presidential election. Nodes represent Twitter users and edges represent retweets during the one-month period before election (Oct. 1 to Nov 1). The graph is laid out using DrL (a force-directed graph layout) and the nodes are assigned different colors and shapes (blue circles and red triangles) according to the two communities detected by the Louvain method \cite{blondel2008fast}. A detailed explanation is in \cite{luo2021echo}.}
	\label{fig:twitter polarized}
\end{figure}

This study aims to control segregation via exogenous incentives, i.e., keep the segregation measure $s$ low by encouraging inter-community connections. 
Specifically, we construct a game which models users' interactions on social networks.
In this game, users are strategic decision makers. When deciding with who to connect,~i.e.,~what edges to form, each user faces a trade off: staying connected to others within the same community vs. obtaining an exogenous incentive for making inter-community connections. 
First, we show that segregation is a natural consequence of the game. Then, we design a mechanism, named \emph{Algorithmic Recommendation Mechanism} (ARM), which is based on the weak-tie theory\cite{granovetter1973strength} and encourages inter-community connections by offering exogenous incentives. With ARM incorporated in the game, we show that connecting with users both in the same community and different community, i.e., integration, is the only rational choice for users in the game to maximize their utilities, which leads to a low segregation measure $s$. In the end of Section \ref{sec:introduction}, some important symbols used in this article are shown in Table 1.  

\subsection{Main Results and Organization}
\noindent(1) By representing the social network as a \emph{directed stochastic block model} (DiSBM)\cite{wilinski2019detectability}, Section \ref{sec: DiSBM} formulates an edge formation game and shows that it leads to a Nash equilibrium where the network is segregated. 

\noindent(2) Section \ref{sec: ARM} proposes an \emph{Algorithmic Recommendation Mechanism} (ARM) which offers users additional rewards by inter-community recommendation. Incorporating this mechanism leads to a Bertrand-like game \cite{milgrom1990rationalizability} with a unique Nash equilibrium where segregation is mitigated.

\noindent(3) Section \ref{sec:markov game} considers the case where ARM's recommendation acceptance probability evolves as a semi-Markov process. Analysis within a stochastic game framework shows that users in the network reach the time-evolving Nash equilibrium of the resulting game.

\noindent(4) Finally, Section \ref{sec: numerical} presents numerical simulations to illustrate how the proposed ARM mitigates segregation and increases inter-community connections. The numerical study also suggests that higher recommendation acceptance probability is effective during polarizing events where the level of segregation is high.

\subsection{Related Work}
Previous works that are related to ours can be considered under three categories:

\noindent(1) {\it Social network segregation:}  
Agent-based opinion dynamics models \cite{sasahara2019inevitability, baumann2020modeling, banisch2019opinion, blex2020positive} are the mainstream approach to study the segregation on social networks. 
Sasahara et al. \cite{sasahara2019inevitability} introduced social influence and unfriending into their model, where users can change both their opinions and connections based on the received information. Baumann et al. \cite{baumann2020modeling} proposed a radicalization mechanism which reinforces extreme opinions from moderate initial conditions. Banisch and Olbrich \cite{banisch2019opinion} considered the social feedback's effect on users expressing alternative opinions and analyzed the sufficient conditions for stable bi-polarization on a stochastic block model-structured network. Blex and Yasseri \cite{blex2020positive} proposed a network-based solution to the Schelling's model and derived that algorithmic bias in the form of rewiring is incapable of preventing segregation. 

In this work, we construct the network as a DiSBM and partition users into two communities (blocks) representing their fixed labels. Similar procedure is considered in \cite{banisch2019opinion, neary2012competing}. More importantly, we adopt a game-theoretic model instead of opinion dynamics to analyze users' decision of forming connections with others, which captures the psychological features of decision-making.

\vspace{0.1cm}
\noindent(2) {\it Nash equilibrium analysis of social networks:} Game theory is a widely used analysis tool for social networks. Jackson and Wolinsky \cite{jackson1996strategic} analyzed the stability and efficiency of social networks when self-interested users can form or cut links in a game setting. Bramoullé et al. \cite{bramoulle2004network} studied the edge formation as individual players choosing their partners in a 2-by-2 anti-coordination games. Avin et al. \cite{avin2018preferential} constructed an evolutionary network formation game and demonstrated that preferential attachment is the unique Nash equilibrium. Alon et al. \cite{alon2010note} formulated a game-theoretic model to deal with the incentives of interested parties outside the network in information diffusion. Mele \cite{mele2017structural} proposed a potential game on a network where user's payoff depends on both directed links and link externalities. Game theoretic models have also been applied to structure identification of industrial cyber-physical systems \cite{zhang2021robust}, conflict mitigation through third party interventions \cite{song2021third}, and trust management mechanism in mobile ad hoc network \cite{li2020exploring}. 

In this work, we propose an edge formation game in which the players (i.e., social network users) face the trade off between connecting within community and forming inter-community connections. The link externalities \cite{mele2017structural} inspire the idea of inter-community friend recommendations in our model, as explained in Section \ref{subsec:algorithmic recommendation}. We analyze the Nash equilibrium of the game to learn whether the network is segregated or not. 

\vspace{0.1cm}
\noindent(3) {\it Weak Ties and segregation mitigation strategies:} In this paper, we propose an \emph{Algorithmic Recommendation Mechanism} (ARM); an overview of the framework is shown in Fig. \ref{fig:algorithmic recommendation}. The ARM incorporates the positive effect of weak ties, or "friends of a friend". Granovetter's work \cite{granovetter1973strength} on Weak Ties-theory demonstrates that a person's weak contacts are more likely to bring novel information such as job opportunities to him compared with his close contacts. Liu et al. \cite{liu2017social} discussed the increased online weak ties with the emergence of new media platforms and functions such as "follow the post", and "retweet". Mele \cite{mele2017structural} proposed a network formation model involving indirect connections, which is an extended version of weak tie. 




\begin{figure}
	\centering
	\includegraphics[width=0.5\textwidth]{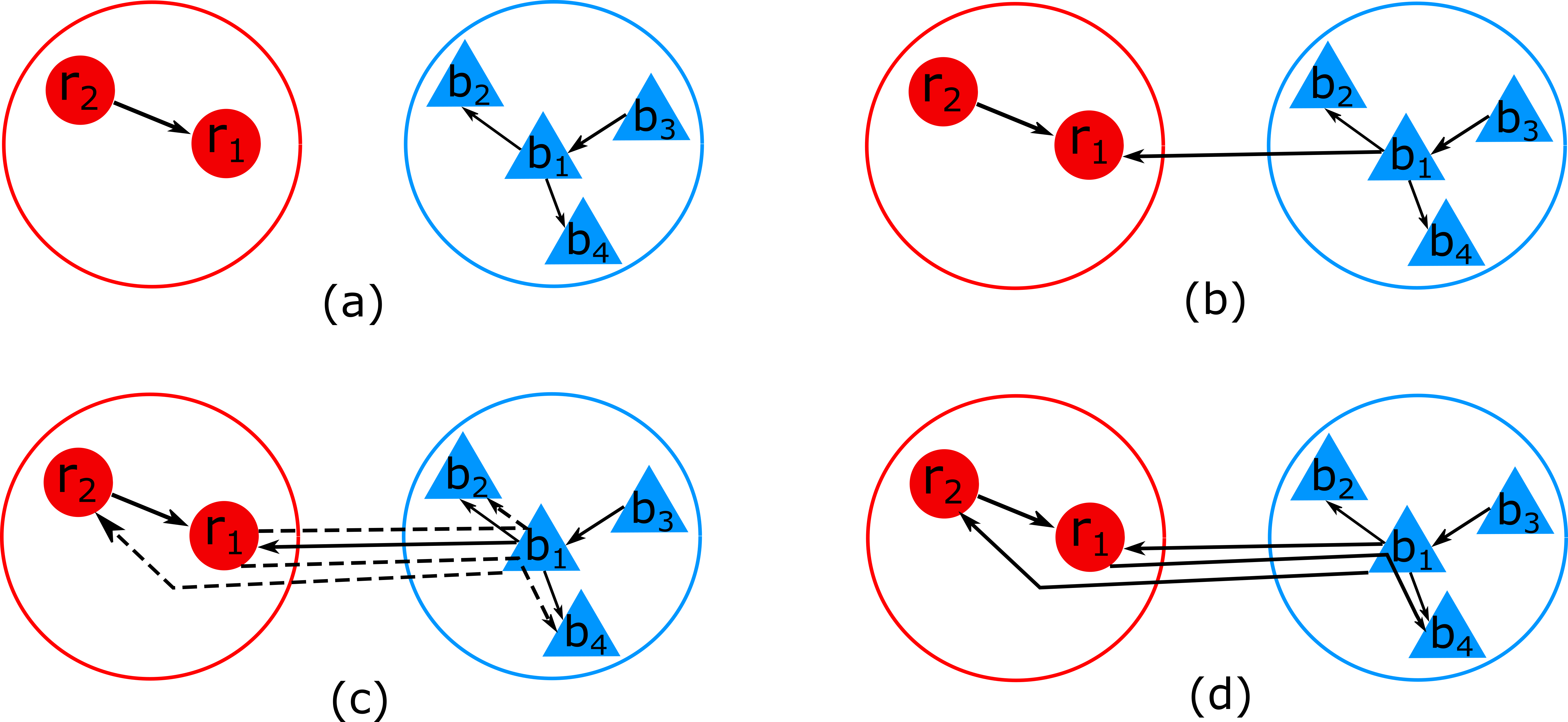}
	\caption{This figure illustrates the proposed \emph{algorithmic recommendation mechanism} (ARM) (details in Algorithm 1). (a) Inside the same community, $b_1$ is followed by $b_2$ and $b_4$, and $r_1$ is followed by $r_2$. (b) Across different communities, $r_1$ follows $b_1$ (i.e., forms an inter-community edge). (c) ARM recommends inter-community links (from $r_1$ to $b_2$ and $b_4$, from $b_1$ to $r_2$) indicated by the three dashed arrows with probability according to (\ref{eq:recommendation prob}). (d) The inter-community recommendations convert into edges (from $b_1$ to $r_2$, from $r_1$ to $b_4$) with some acceptance probability (defined as $C$ in Algorithm 1) and increase users' utility. Therefore, ARM incentivizes users to proactively form inter-community edges and mitigates segregation in social networks.}
	\label{fig:algorithmic recommendation}
\end{figure}

\begin{table} \label{table:symbols}
  \caption{Glossary of Symbols Used in This Article}
  \centering 
  \begin{threeparttable}
    \begin{tabular}{ C{0.05\textwidth}L{0.4\textwidth} } 
    \toprule
    Symbols  & Description \\ 
    \midrule\midrule
    $G^{(t)}$ & the directed graph representing the network at time $t$ \\
    $V$ & the fixed set of users in the network \\
    $E^{(t)}$ & the set of directed edges at time $t$ \\
    $R$ & the fixed set of red users \tnote{*} \\
    $B$ & the fixed set of blue users \\
    $N$ & the number of red users, which also equals the number of blue users \\
    $p_R^{(t)}$ & the probability that a red user follows\tnote{**}  another red user at time $t$ \\
    $p_B^{(t)}$ & the probability that a blue user follows another blue user at time $t$ \\
    $P^{(t)}$ & the probability matrix of the DiSBM corresponding to the network \\
    $U_i^{(t)}$ & user $i$'s utility at time $t$ in the original game \\
    $\arm$ & user $i$'s utility at time $t$ in the game with ARM \\
    $D^{t}(i,j)$ & whether user $j$ follows user $i$ from the different community at time $t$ \\
    $S^{t}(i,j)$ & whether user $j$ follows user $i$ from the same community at time $t$ \\
    $C$ & the probability that a recommendation by ARM is accepted \\
    \midrule\midrule
    \end{tabular}
    \begin{tablenotes}
\item[*] "Red users" is abbreviation for the users from the red community. So is "blue users".
\item[**] Throughout the table and the rest of this paper, we use "follow" interchangeably with "connect with" to denote initiating a directed edge with others, with the edge from the "friend" pointing to the initiator (i.e., "follower").
\end{tablenotes}
\end{threeparttable}
  \end{table}

\section{Directed Network Edge Formation Game}
\label{sec: DiSBM}
In this section, we first propose an edge formation protocol followed by users in the social network. This protocol results in a \emph{directed stochastic block model} (DiSBM) where, the parameters correspond to the actions of users. The edge formation protocol is based on a best response strategy followed by users~i.e.,~users respond to what others did in the previous time instant in order to maximize a utility function. By analyzing the game that corresponds to the best response-based edge formation protocol, our main result of the section shows that it has a unique Nash equilibrium that corresponds to a segregated network~(network with two disconnected communities), and the best response-based edge formation protocol converges to this Nash equilibrium.


\subsection{DiSBM Based Edge Formation Protocol}
\label{subsec: SBM intro}
 
In social networks, users are naturally partitioned into communities based on nodal attributes such as geographic locations, party affiliations, and personal interests \cite{wang1987stochastic}. For example, \cite{conover2011political} demonstrates that the network of political communication on Twitter exhibits a highly segregated partisan structure. The \emph{directed stochastic block model} (DiSBM) \cite{wilinski2019detectability} is frequently used to represent the structure of such networks with communities. 
 
We consider a time-varying version of the DiSBM model with two communities ($N$ red users and $N$ blue users\footnote{We assume equal numbers of red and blue users to simplify the expressions and analysis in the following sections, and this assumption can be relaxed easily.}) where the model parameters correspond to the actions taken by users in each community in a repeated game (i.e., a normal form game which is repeated via a best response strategy adopted by the players). More specifically, at each time instant $t$, user $i$ chooses the edge formation probabilities (i.e., its actions) in a manner that maximizes the expected value of the utility function,
\begin{equation}\label{eq:utility_function}
    U_{i}^{(t)} =  \displaystyle \sum_{j\in R\cup B} \big[D^{(t)}(i, j) - D^{(t)}(j, i) \big].
\end{equation}
where, 
\begin{equation}
\label{eq: D function}%
\resizebox{.91\hsize}{!}{$
    D^{(t)}(i, j) = 
    \begin{cases}
     1 & (i,j) \in E^{(t)} \textrm{ and } i,j\textrm{ in different communities} \\
    0 & (i,j) \notin E^{(t)} \textrm{ or } i,j\textrm{ in the same community.}
    \end{cases}
    $}
\end{equation}

The utility function $U_{i}^{(t)}$ consists of two components. The first component $\displaystyle \sum_{j\in R\cup B} D^{(t)}(i, j)$ is $i$'s number of followers in a different community~i.e.,~its popularity in another community. The second component $\displaystyle \sum_{j\in R\cup B} D^{(t)}(j, i)$ is $i$'s number of friends in different community, representing its efforts in maintaining inter-community connections\footnote{A similar formulation of social network users' utility functions was used in \cite{avin2017assortative}, where the author defines the cut size (i.e., the number of inter-community edges) as part of the utility function. The corresponding network is undirected.}. 
Thus, the utility function (\ref{eq:utility_function}) represents how users face a trade-off between popularity in different community and homophily in its own community when forming connections in social networks.

With this notation, the DiSBM based edge formation via utility maximization is given in Protocol~1.

\noindent\rule{0.5\textwidth}{1pt}
\begin{spacing}{0.55}
\setlength{\parindent}{0pt}
\textbf{Protocol 1.} DiSBM Based Network Edge Formation
\end{spacing}
\noindent\rule{0.5\textwidth}{1pt}

\noindent
{\bf Input: }
$G^{(t)} \textrm{=} \{V, E^{(t)}\}$ where $V \textrm{=} R \cup B$, $t = 0,1,2,\cdots$

\noindent
{\bf Output: }
$p_R^{(t)}, p_B^{(t)}$

\noindent
{\bf Process: }
\begin{enumerate}[\hspace{0in}1)]
\item $p_R^{(0)}, p_B^{(0)} \sim \textrm{Unif}[0,1]$ (i.e.,~the initial parameters of the model are sampled from a uniform distribution).

\item At each odd time instant (i.e.,~$t=1,3,5,\cdots$), red users take actions according to steps 2.1 and 2.2 below while the blue users adhere to the action they adopted at time $t-1$. 
\begin{itemize}
    \item[2.1)] $\forall i,j\in R$, $i$ connects with $j$ with probability\footnote{As will be discussed in the proof of Theorem \ref{th:woARM} (Appendix \ref{ap:th1}), $\mathbb{E}\left\{U_i^{(t)}\right\}$ is identical for any $i\in R$ and is a function of $p_R^{(t)}$.}
\begin{equation}
{p_R^{(t)}} = \displaystyle \argmax_{p_R^{(t)}\in (0,1]} \mathbb{E}\left\{U_i^{(t)}\right\}
\end{equation}
where $\mathbb{E}$ is the expectation with respect to the probability distribution induced by the DiSBM model, and $U_i^{(t)}$ is the utility function defined in (\ref{eq:utility_function}).


\item [2.2)] $\forall i\in R, j\in B$, $i$ connects with $j$ with probability $\frac{1-p_{R}^{(t)}}{N}$.  
\end{itemize}

\item At each even time instant (i.e.,~$t=2,4,6,\cdots$), blue users take actions according to steps 3.1 and 3.2 below while the red users adhere to the action they adopted at time $t-1$. 
\begin{itemize}
    \item[3.1)] $\forall i,j\in B$, $i$ connects with $j$ with probability
\begin{equation}
{p_B^{(t)}} = \displaystyle \argmax_{p_B^{(t)}\in (0,1]} \mathbb{E}\left\{U_i^{(t)}\right\}
\end{equation}

\item [3.2)] $\forall i\in B, j\in R$, $i$ connects with $j$ with probability $\frac{1-p_{B}^{(t)}}{N}$.  
\end{itemize}
\end{enumerate}

\enlargethispage{\baselineskip}
\noindent\rule{0.5\textwidth}{1pt}

\newtheorem*{remark}{Remark 1}
\subsection{Discussion of Protocol 1}
\label{subsec:protocol}
\noindent
{\bf Protocol 1 induces a DiSBM characterised by best response strategy: }

\noindent In Protocol 1, red users and blue users take actions alternatively by maximizing utility functions. More specifically, at each odd time instant $t$, red users choose the probability ${p_R^{(t)}}$ via maximizing the expected utility of a sample uniformly drawn from the red community $R$, while blue users keep the probability ${p_B^{(t)} = {p_B^{(t-1)}}}$. Blue users choose their parameters at even time steps according to similar steps.
Therefore, Protocol 1 can be viewed as a DiSBM with time varying parameters resulting from the best response strategies played by users in the network. In other words, Protocol 1 corresponds to a DiSBM with the probability matrix
\begin{equation}
\label{eq:SBM prob}
    P^{(t)} = \begin{bmatrix}
p_{R}^{(t)} & \frac{1-p_{R}^{(t)}}{N} \\
\frac{1-p_{B}^{(t)}}{N} & p_{B}^{(t)}
\end{bmatrix}
\end{equation}
where $p_R^{(t)}, p_B^{(t)}$ are the best response strategies taken by the red and blue users at each time instant, respectively.

\vspace{0.2cm}
\noindent
{\bf Protocol 1 corresponds to the regime with dense intra- and sparse inter-community edges: }

\noindent In the context of real world social networks, Protocol 1 corresponds to the regime with sparse inter-community edges and dense intra-community edges.\footnote{Note that $P^{(t)}$ does not need to be symmetric considering that the edges are directed; also the rows or columns in $P^{(t)}$ (given in (\ref{eq:SBM prob})) do not need to sum to $1$.} In other words, the the probability matrix $P^{(t)}$ is of the form,
\begin{equation}
\label{eq: SBM structure}
    P_{ab}^{(t)} =
    \begin{cases}
    O(1/N), & a\neq b \\
    O(1), & a=b
    \end{cases}
\end{equation}
A connection probability matrix of the form (\ref{eq: SBM structure}) emulates the structure of many real world social networks where individuals  are densely connected within a community (i.e.,~$O(N^2)$ edges within a community of size $N$) and sparsely connected between communities (i.e.,~$O(N)$ edges between two communities each of size $N$).

Connection probability matrices of the form (\ref{eq: SBM structure}) are further motivated by the fact that users act differently when forming intra-community and inter-community edges in social networks. For example, Twitter users have different political opinions and tend to follow or retweet users with similar opinions. Sports forum (e.g. Reddit) users have different favorite teams and tend to communicate more frequently with users who support the same team. 

\begin{remark}[\emph{The game corresponding to the protocol}]
\label{rm:game&protocol}
\normalfont
Note that 
there exists a normal-form game corresponding to Protocol~1 where, set of nodes $V=R\cup B$ are the players, the interval $(0,1]$ is the set of actions for each player $i \in V$ and, the utility function of each $i \in V$ is
$\mathbb{E}_{j\in R}\left\{ U_{j}^{(t)}\right\}$ or $\mathbb{E}_{j\in B}\left\{ U_{j}^{(t)}\right\}$, i.e., the expected utility function of users in one community depending on the community $i$ belongs to (where $U_{j}^{(t)}$ is given in (\ref{eq:utility_function}) and $t$ is any fixed time instant). When this normal-form game is of a special type (e.g.~a strictly dominant strategy profile, a submodular game), the best response based Protocol~1 is guaranteed to converge to the game's Nash equilibrium. Henceforth, for any such best response based protocol, we use the term \emph{the game corresponding to the protocol} to refer to this induced normal-form game by that protocol.
\end{remark}

\newtheorem{theorem}{Theorem}
\subsection{Nash Equilibrium Analysis of the Game Corresponding to Protocol 1}
\label{subsec: game analysis}
Thus far (in Sec.~\ref{subsec: SBM intro} and \ref{subsec:protocol}) we discussed the intuition behind the DiSBM based edge formation (Protocol~1) in which the network parameters arise from each community maximizing a utility function~(i.e.,~a best response strategy). In this subsection, we analyze the game that corresponds to Protocol~1~(recall from Remark~1 that this is the induced normal-form game) and the main result of this subsection (Theorem~\ref{th:woARM}) indicates that:
\begin{enumerate}[i.]
    \item Protocol~1 converges to a stationary state i.e.,~the sequence of parameters ($p_{R}^{(t)}, p_{B}^{(t)}, t = 1,2,..$) chosen by individuals from each community converges to fixed values $p_{R}, p_{B}$;
    
    \item the fixed values at the stationary state are $p_{R}^{(t)} =1, p_{B}^{(t)} = 1$ and they correspond to the unique Nash equilibrium of the normal-form game corresponding to Protocol~1~i.e.,~at the stationary state, the network will almost surely have no inter-community edges, leading to echo chambers.
\end{enumerate}


\begin{theorem}[Convergence of Protocol~1 to the Nash Equilibrium]
\label{th:woARM}
Consider the best response dynamics given in Protocol~$1$ (Sec.~\ref{subsec: SBM intro}). Segregation (i.e., $p_{R}^{(t)}\textrm{=} p_{B}^{(t)}\textrm{=} 1$) is the unique Nash equilibrium of the corresponding game and $p_{R}^{(t)}, p_{B}^{(t)}$ both converge to it as time $t$ tends to infinity.
\end{theorem}

\begin{proof}
The proof of Theorem \ref{th:woARM} is straightforward with the main idea being that the Nash equilibrium ($p_{R}^{(t)}\textrm{=} p_{B}^{(t)}\textrm{=} 1$) corresponds to users' strictly dominant strategy. The detailed argument is given in Appendix~\ref{ap:th1} for completeness.
\end{proof}

\section{How Algorithmic Recommendation Mechanism (ARM) Reshapes the Segregation Equilibrium}
\label{sec: ARM}

Sec.~\ref{sec: DiSBM} showed that segregation is the Nash equilibrium of the game corresponding to Protocol~1. This leads us to the following question: how can we augment Protocol 1 such that, 
\begin{enumerate}[i.]
    \item the normal-form game corresponding to the augmented protocol has a Nash equilibrium that is not segregation;
    
    \item the strategy profiles of the players converge to the Nash equilibrium (in the previous point) over time. 
\end{enumerate}
To this end, this section presents an algorithmic recommendation mechanism (ARM),  which incentivizes users to form inter-community edges. 
Then, we illustrate how ARM changes the Nash equilibrium of the corresponding game (compared with the segregation result shown in Theorem 1) and mitigates segregation in the network. 
\subsection{Algorithmic Recommendation Mechanism (ARM)}
\label{subsec:algorithmic recommendation}

The aim of introducing ARM is to reshape the Nash equilibrium of the game corresponding to Protocol~1 (which leads to a segregated network as shown in Sec.~\ref{subsec: game analysis}) such that the resulting network is not segregated.
To this end, ARM incentivizes users to form inter-community edges by providing exogenous rewards.
More specifically, if $i$ follows $j$ in the different community, and $j$ has not yet followed $i$, ARM will recommend $i$ to $j$ with a probability proportional to the number of 2 hop connections from $i$ to $j$. Then, the link recommendation is accepted by $j$ with some probability, thus increasing $i$'s popularity in another community and its utility (as indicated in the utility function (\ref{eq:utility_function})). 
Fig. \ref{fig:algorithmic recommendation} shows an example of ARM recommending links.

Recall the function $D^{(t)}$ defined in (\ref{eq: D function}) to count the number of edges between different communities. Similarly, we define another function: 
\begin{equation}
\label{eq: S function}
\resizebox{.91\hsize}{!}{$
    S^{(t)}(i, j) = 
    \begin{cases}
    1 & (i,j) \in E^{(t)} \textrm{ and } i,j\textrm{ in the same community}  \\
    0 & (i,j) \notin E^{(t)} \textrm{ or } i,j\textrm{ in different communities.}
    \end{cases}
    $}
\end{equation}
Thus, for a given pair of nodes $i, j$, $S^{(t)}$ indicates whether an edge $(i,j)$ exists and whether it is between the same community. With this notation, we are now equipped to present the ARM.

\begin{figure}
	\centering
	\includegraphics[width=0.35\textwidth]{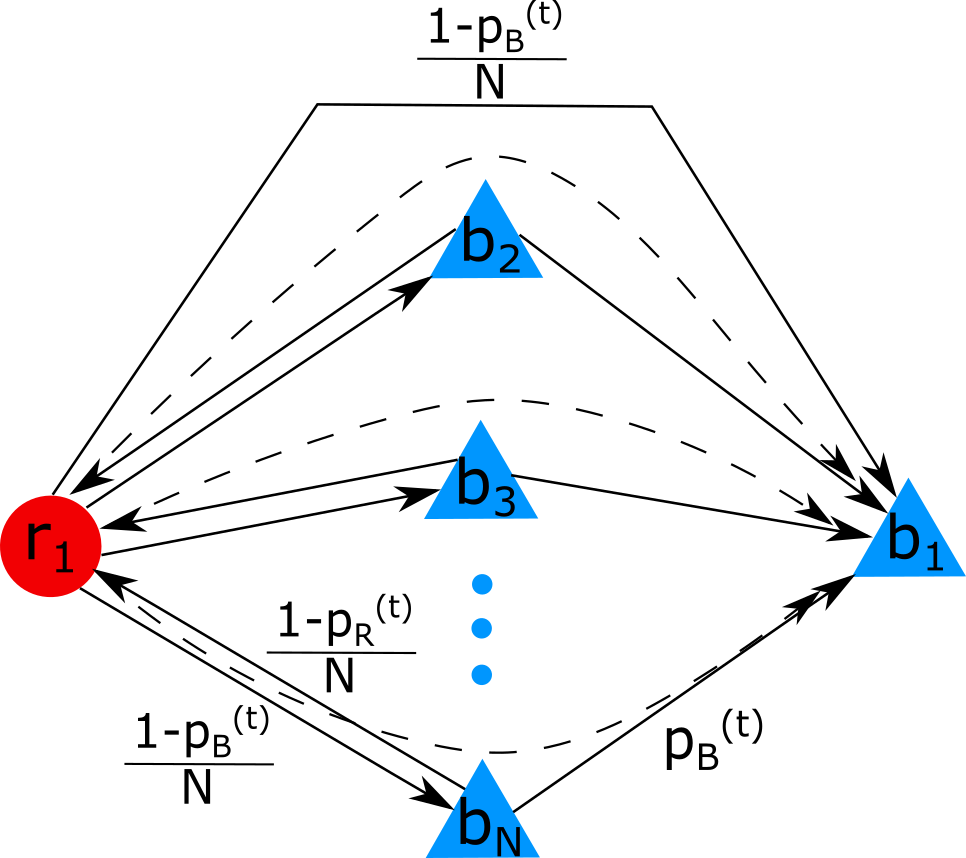}
	\caption{This figure illustrates how ARM (Sec. \ref{subsec:algorithmic recommendation} Algorithm $1$) recommends an inter-community link from red user $r_1$ to blue user $b_1$ based on the existing 2 hop connections. The dashed arrow curves denote the possible 2 hop connections from $r_1$ to $b_1$. 
	The solid arrow line on top denotes the probability that $b_1$ already follows $r_1$ (i.e., edge $(r_1, b_1)$ already exists), in which case ARM will not recommend $r_1$ to $b_1$. If edge $(r_1, b_1)$ does not exist yet, ARM will recommend $r_1$ to $b_1$ with a probability proportional to the number of 2 hop connections between $r_1$ and $b_1$. 
	Labels on solid arrow lines denote corresponding edges' formation probabilities.}
	\label{fig:possible weak ties}
\end{figure}

\noindent\rule{0.5\textwidth}{1pt}
\begin{spacing}{0.65}
\setlength{\parindent}{0pt}
\textbf{Algorithm 1.} Algorithmic Recommendation Mechanism (ARM)
\end{spacing}
\noindent\rule{0.5\textwidth}{1pt}
{\bf Input: } $E^{(t)}$ -- the edge set at time $t$ and $R,B$ -- the two node sets representing two communities.

\noindent
{\bf Output: } $\{(i, j)\}$ -- the set of recommendation links.

\noindent
{\bf Process: }
At time $t$, for any ordered pair $(i,j)$ such that $i,j$ are in different communities, ARM forms an edge from $i$ to $j$ according to steps below:
\begin{enumerate}[\hspace{0in}1)]
  \item If $j$ follows $i$, i.e., $(i, j)\in E^{(t)}$, then stop recommendation; otherwise proceed to step $2$.
  \item 
  For any user $j'$ such that $S^{(t)}(j', j)=1$
  , count the total number of edges between such $j'$ and $i$, i.e., $\sum_{j'} D^{(t)}(i, j') + D^{(t)}(j', i)$. 
Recommend $i$ to $j$ with probability
  \begin{equation}
  \label{eq:recommendation prob}
      \frac{\displaystyle \sum_{j'\in R\cup B}  \big[D^{(t)}(i, j') + D^{(t)}(j', i)\big] 
  S^{(t)}(j', j)}{\textstyle N-1}
  \end{equation}
  where
  \begin{equation}
  \label{eq:2 hop}
      \displaystyle \sum_{j'\in R\cup B}  \big[D^{(t)}(i, j') + D^{(t)}(j', i)\big] 
  S^{(t)}(j', j)
  \end{equation} is the number of 2 hop connections from i to j, and $N-1$ is used to offset the linear growth of 2 hop connections (between two users) to the number of users in the network.
  
  \item If $i$ was recommended to $j$ in step 2, 
  $j$ accepts the recommendation (i.e., the edge $(i, j)$ is formed) according to a Bernoulli random variable with probability of success $C$, where $C$ is the \emph{acceptance probability}.
  
\end{enumerate}

\noindent\rule{0.5\textwidth}{1pt}

\vspace{0.2cm}
\noindent
{\bf Discussion of Algorithm $1$: }

\noindent (1) Step $1$ considers the fact that an edge ARM recommends may already exist. In that case, no further steps are taken by ARM. 

\noindent (2) Step $2$ recommends inter-community links with probability according to (\ref{eq:recommendation prob}), which favors recommending links between users with more 2 hop connections. To justify how 2 hop connections increase the recommendation probability, suppose $i$ is a red user and $j$ a blue user, then as indicated in (\ref{eq:2 hop}), $\displaystyle \sum_{j'\in R\cup B} D^{(t)}(i, j') S^{(t)}(j', j)$ specifies the case that $i$ is followed by $j$'s blue friends; $\displaystyle \sum_{j'\in R\cup B} D^{(t)}(j', i) S^{(t)}(j', j)$ specifies the case that $i,j$ share blue friends in common. In both cases, the bigger the value is, the more likely that ARM will recommend $i$ to $j$.


\noindent (3) Step $3$ considers the fact that each recommended link turns into an actual edge according to a Bernoulli random variable. In practice, the probability of success depends on the amount of incentive provided by the network administrator in the form of exogenous reward.

\noindent (4) We provide the complexity analysis of ARM as follows. Step 2 queries the union of $j$'s follower list and $i$'s follower list/friend list. With hash table based data structure, the query takes $O(1)$ time and the space complexity of the hash table is $O(N)$. The overall time complexity and space complexity are both $O(N^2)$. There are specialized data structure for real-world large-scale graphs which is out of scope in this work, and we refer interested readers to a survey of graph database models\cite{angles2008survey}.

\subsection{Edge Formation Protocol with ARM}
\label{subsec: protocol 2}
To explore how ARM affects segregation, we first augment the utility function (\ref{eq:utility_function}) in a manner that reflects the trade-off between staying connected within the same community and obtaining exogenous incentives provided by ARM. The augmented utility function is
\begin{equation}\label{eq:utility_function_2}
\resizebox{.89\hsize}{!}{$
\begin{aligned}
    \scriptstyle \arm = & \displaystyle \sum_{j\in R\cup B} \big[ \scriptstyle D^{(t)}(i, j) - D^{(t)}(j, i) \big] 
    + \displaystyle \sum_{j\in R\cup B} \Big\{\scriptstyle C \big[\scriptstyle 1-D^{(t)}(i, j) \big] \\
    & \frac{\displaystyle \sum_{j'\in R\cup B} \scriptstyle \big[D^{(t)}(i, j') + D^{(t)}(j', i)\big] 
    \scriptstyle S^{(t)}(j', j)}{\scriptstyle N-1} \Big\} \\
    & + \displaystyle \sum_{j\in R\cup B} \displaystyle \sum_{i'\in R\cup B} \Big\{\scriptstyle C \big[\scriptstyle 1-D^{(t)}(j, i') \big] \scriptstyle \big[\frac{S^{(t)}(i, i') D^{(t)}(j, i)}{\scriptstyle N-1} \big] \Big\}.
\end{aligned}
$}
\end{equation}
In (\ref{eq:utility_function_2}), the first summation term is the original utility as in (\ref{eq:utility_function}). The second summation term 
\begin{equation}
\resizebox{.89\hsize}{!}{$
    \displaystyle \sum_{j\in R\cup B}
    \Big\{\scriptstyle C \big[\scriptstyle 1-D^{(t)}(i, j) \big] \\
    \frac{\displaystyle \sum_{j'\in R\cup B} \scriptstyle \big[D^{(t)}(i, j') + D^{(t)}(j', i)\big] 
    \scriptstyle S^{(t)}(j', j)}{\scriptstyle N-1} \Big\}
    $}
\end{equation}
is the expected number of inter-community edges formed by ARM for user $i$. The third summation term 
\begin{equation}
    \displaystyle \sum_{j\in R\cup B} \displaystyle \sum_{i'\in R\cup B} \Big\{\scriptstyle C \big[\scriptstyle 1-D^{(t)}(j, i') \big] \scriptstyle \big[\frac{S^{(t)}(i, i') D^{(t)}(j, i)}{\scriptstyle N-1} \big] \Big\}    
\end{equation}
is the reward for connecting users from two communities (e.g. a red user connects its red followers to its blue friends) in ARM recommendations. These two summation terms represent ARM's exogenous incentives, where the users' benefit of connecting with others depends on the composition of friends of friends\footnote{A similar assumption was used in \cite{mele2017structural} and \cite{de2009ethnic}, where the value of indirect links (i.e., weak ties) affects users’ cost of linking.}.


The edge formation via maximizing the augmented utility (\ref{eq:utility_function_2}) is given in Protocol $2$ below.


\noindent\rule{0.5\textwidth}{1pt}
\begin{spacing}{0.55}
\setlength{\parindent}{0pt}
\textbf{Protocol 2.} DiSBM Based Network Edge Formation with ARM
\end{spacing}
\noindent\rule{0.5\textwidth}{1pt}
\noindent
{\bf Input: }
$G^{(t)} \textrm{=} \{V, E^{(t)}\}$ where $V \textrm{=} R \cup B$, $t = 0,1,2,\cdots$

\noindent
{\bf Output: }
$p_R^{(t)}, p_B^{(t)}$

\noindent
{\bf Process: }
\begin{enumerate}[\hspace{0in}1)]
\item $p_R^{(0)}, p_B^{(0)} \sim \textrm{Unif}[0,1]$.

\item At each odd time instant (i.e.,~$t=1,3,5,\cdots$), red users take actions according to steps 2.1 and 2.2 below while the blue users adhere to the action they adopted at time $t-1$. 
\begin{itemize}
    \item[2.1)] $\forall i,j\in R$, $i$ connects with $j$ with probability\footnote{As will be discussed in the proof of Theorem \ref{th:wARM} (Appendix \ref{ap:th2}), $\mathbb{E}\left\{\arm\right\}$ is identical for any $i\in R$ and is a function of $p_R^{(t)}$.}

\begin{equation}
{p_R^{(t)}} = \displaystyle \argmax_{p_R^{(t)}\in (0,1]} \mathbb{E}\left\{\arm\right\}
\end{equation}
where $\mathbb{E}$ is the expectation with respect to the probability distribution induced by the DiSBM model, and $\arm$ is the augmented utility function defined in (\ref{eq:utility_function_2}).

\item [2.2)] $\forall i\in R, j\in B$, $i$ connects with $j$ with probability $\frac{1-p_{R}^{(t)}}{N}$. 
\end{itemize}

\item At each even time instant (i.e.,~$t=2,4,6,\cdots$), blue users take actions according to steps 3.1 and 3.2 below while the red users adhere to the action they adopted at time $t-1$. 
\begin{itemize}
    \item[3.1)] $\forall i,j\in B$, $i$ connects with $j$ with probability
\begin{equation}
{p_B^{(t)}} = \displaystyle \argmax_{p_B^{(t)}\in (0,1]} \mathbb{E}\left\{\arm\right\}
\end{equation}

\item [3.2)] $\forall i\in B, j\in R$, $i$ connects with $j$ with probability $\frac{1-p_{B}^{(t)}}{N}$.  
\end{itemize}

\item At each time instant, ARM forms inter-community edges according to Algorithm $1$.
\end{enumerate}

\enlargethispage{\baselineskip}
\noindent\rule{0.5\textwidth}{1pt}  

\subsection{Nash Equilibrium Analysis of the Game with ARM}
\label{subsec: game analysis 2}

In the previous subsection (Sec. \ref{subsec: protocol 2}), we proposed Protocol~2 where users play the best response in a network incorporated with ARM. In this subsection, we analyze the game corresponding to Protocol~2~(recall Remark~1) and show that:
\begin{enumerate}[i.]
    \item the game has a submodular structure similar to the Bertrand game \cite{milgrom1990rationalizability}, which guarantees that Protocol 2 converges to a fixed value $(p_R,p_B)$; 
    \item the fixed value at the steady state depends on the \emph{acceptance probability} $C$ (defined in Algorithm $1$). If $C>\frac{1}{2}$, the Nash equilibrium of the corresponding game leads to social integration, i.e., the only rational choice for users to maximize their utilities is to connect with both users in the same community and different community.
\end{enumerate}


\begin{theorem}[Convergence of Protocol 2 to the Nash Equilibrium]
\label{th:wARM}
Consider the best response dynamics given in Protocol 2 (Sec. \ref{subsec: protocol 2}). If the acceptance probability (i.e., $C$ defined in Algorithm 1) is greater than $\frac{1}{2}$, then segregation is mitigated (i.e., $p_{R}^{(t)}<1,\,p_{B}^{(t)}<1$) at the unique Nash equilibrium of the corresponding game, and $p_{R}^{(t)}, p_{B}^{(t)}$ both converge to the Nash equilibrium as time $t$ tends to infinity.
\end{theorem}
\begin{proof}
The main idea behind the proof of Theorem \ref{th:wARM} is to show that the game corresponding to Protocol 2 is submodular, which guarantees the convergence of users' best response dynamics to the Nash equilibrium.  The detailed proof is given in Appendix~\ref{ap:th2}.
\end{proof}
Theorem \ref{th:wARM} illustrates how the proposed ARM reshapes the segregation equilibrium of the game and results in an integrated network, where users are incentivized to form inter-community connections.

\section{Stochastic Game with Markovian ARM Parameter}
\label{sec:markov game}
Recall that in Protocol 2 (Sec. \ref{subsec: protocol 2}) we assume ARM's link recommendation is accepted according to a Bernoulli random variable with fixed \emph{acceptance probability} $C$. In this section we relax the assumption by allowing $C$ to evolve over time as the sample path of a semi-Markov process. We then modify Protocol 2 to account for this time-variant effect, and illustrate that the modified protocol (Protocol 3) induces best response dynamics which tracks the Nash equilibrium evolving with the time-variant $C$. 

\subsection{Edge Formation Protocol with Markovian Acceptance Probability}
\label{subsec: protocol 3}
\noindent
{\bf Motivation: }
Real world social networks have time-variant level of segregation. We justify this phenomenon by the two motivating examples used in Sec. \ref{subsec:protocol} as follows. During a politically polarizing event (e.g. election or legislation), Twitter users tend to segregate and reduce their connections with others from a different community. Similarly, on sports forums such as Reddit's related channels, users' online communication tend to be more polarized during sports league's season (e.g. NBA's playoff or NCAA's March Madness), while less polarized during off-season. In the literature, \cite{naskar2019modelling} models Twitter users' emotion dynamics as a Markov process. \cite{luo2021echo} characterizes the pattern of Twitter users' retweets during the presidential election using a Markov bridge model. \cite{iyengar2011content, sakaki2010earthquake} utilize Twitter users' messages as signals to a hidden Markov model to determine when an anticipated event (e.g. social activities, sports, weather) starts. 

The seasonal pattern of segregation requires the network administrator to spend time-varying efforts to control it, i.e., spend different efforts on recommendation at different times to encourage inter-community connections. We capture this temporal pattern by modeling the \emph{acceptance probability} $C$ (defined in Algorithm $1$) as a semi-Markov process $C^{(t)}, t=0,1,\cdots$. 

The semi-Markov process of the \emph{acceptance probability} is Markovian only at specified jump instants, i.e., when $C^{(t)}$ transits. As $C^{(t)}$ transits into the next state, it stays there for a state holding time $T_h$, which we assume is a constant. $T_h$ measures the interval between two consecutive transitions of $C^{(t)}$ (i.e., how often the network administrator changes its operation condition), thus it is much longer than (e.g. 100$\times$) the time scale of user's action in the protocol. More precisely, we define the semi-Markov process as follows:
\vspace{0.2cm}

$C^{(t)}$ takes value in a finite state space $S_C=\{C_1, C_2,\cdots, C_n\}$, has an initial state $C^{(0)}\in S_C$, and a Markov transition probability matrix $P$ conditionally independent of users' actions
\begin{equation}
\label{eq: evolve C}
\resizebox{.89\hsize}{!}{$
\begin{aligned}
&P(C^{(t+1)}=j|C^{(t)}=i,p_R^{(t)}, p_B^{(t)})
=P(C^{(t+1)}=j|C^{(t)}=i) \\
&=
\begin{cases}
P_{ij},  &t+1 \textrm{ mod } T_h = 0 \\
1, &t+1 \textrm{ mod } T_h \neq 0 \textrm{ and }i=j\\
0, &t+1 \textrm{ mod } T_h \neq 0 \textrm{ and }i\neq j \\ 
\end{cases}    
\end{aligned}
$}
\end{equation}
where $p_{R}^{(t)}, p_{B}^{(t)}$ are respectively red and blue users' actions at time $t$. 
With these notations, the protocol with Markovian \emph{acceptance probability} is given in Protocol 3.

\noindent\rule{0.5\textwidth}{1pt}
\begin{spacing}{0.55}
\setlength{\parindent}{0pt}
\textbf{Protocol 3.} DiSBM Based Network Edge Formation with Markovian Acceptance Probability
\end{spacing}
\noindent\rule{0.5\textwidth}{1pt}
\noindent
{\bf Input: }
$G^{(t)} \textrm{=} \{V, E^{(t)}\}$ where $V \textrm{=} R \cup B$, $t = 0,1,2,\cdots$

\noindent
{\bf Output: }
$p_R^{(t)}, p_B^{(t)}$

\noindent
{\bf Process: }
\begin{enumerate}[\hspace{0in}1)]
\item $p_R^{(0)}, p_B^{(0)} \sim U[0,1]$; $C^{(0)}$ is the initial \emph{acceptance probability}. 


\item At each odd time instant (i.e.,~$t=1,3,5,\cdots$), red users take actions according to steps 2.1 and 2.2 below while the blue users adhere to the action they adopted at time $t-1$. 
\begin{itemize}
    \item[2.1)] $\forall i,j\in R$, $i$ connects with $j$ with probability
\begin{equation}
{p_R^{(t)}} = \displaystyle \argmax_{p_R^{(t)}\in (0,1]} \mathbb{E}\left\{\arm\right\}
\end{equation}
where $\mathbb{E}$ is the expectation with respect to the probability distribution induced by the DiSBM model, and $\arm$ is the augmented utility function defined in (\ref{eq:utility_function_2}).

\item [2.2)] $\forall i\in R, j\in B$, $i$ connects with $j$ with probability $\frac{1-p_{R}^{(t)}}{N}$. 
\end{itemize}

\item At each even time instant (i.e.,~$t=2,4,6,\cdots$), blue users take actions according to steps 3.1 and 3.2 below while the red users adhere to the action they adopted at time $t-1$. 
\begin{itemize}
    \item[3.1)] $\forall i,j\in B$, $i$ connects with $j$ with probability
\begin{equation}
{p_B^{(t)}} = \displaystyle \argmax_{p_B^{(t)}\in (0,1]} \mathbb{E}\left\{\arm\right\}
\end{equation}

\item [3.2)] $\forall i\in B, j\in R$, $i$ connects with $j$ with probability $\frac{1-p_{B}^{(t)}}{N}$.  
\end{itemize}





\item At each time instant, $C^{(t)}$ evolves according to the transition rule (\ref{eq: evolve C}). ARM forms inter-community edges according to Algorithm $1$ with \emph{acceptance probability} $C^{(t)}$.
\end{enumerate}
\enlargethispage{\baselineskip}
\noindent\rule{0.5\textwidth}{1pt}

\subsection{Nash Equilibrium Analysis of the Game Corresponding to Protocol 3}
\label{subsec:infinite horizon game with varying C}
In Protocol 3, we consider the edge formation process in an infinite horizon where users aim to maximize their discounted sum of payoff. We study Protocol 3 by analyzing a corresponding stochastic game, with the Markovian \emph{acceptance probability} $C^{(t)}$ as the state of the game. The main result of this subsection (Theorem 3) indicates that:
\begin{enumerate}[i.]
    \item Protocol 3 captures user's best response dynamics, which converges to a steady state dependent on $C^{(t)}$. Once $C^{(t)}$ transits to a new state, user's best response will converge to a new steady state;
    \item user's optimal strategy is to myopically optimize its one-stage payoff at each round of the game (with $C^{(t)}$ as the state of the game).
\end{enumerate}

\begin{theorem}[Convergence of Protocol 3 to the time-varying Nash Equilibrium]
\label{th:wARM2}
Consider the best response dynamics given in Protocol 3 (Sec. \ref{subsec: protocol 3}). If the acceptance probability (i.e., $C^{(t)}$) evolves as a semi-Markov process specified in (\ref{eq: evolve C}), then $p_{R}^{(t)}, p_{B}^{(t)}$ both converge to the Nash equilibrium dependent on $C^{(t)}$ before $C^{(t)}$ transits.
\end{theorem}
\begin{proof}
The proof of Theorem \ref{th:wARM2} is based on the assumption that the transition dynamics of $C^{(t)}$ are conditionally independent of users' actions as shown in (\ref{eq: evolve C}). Therefore before $C^{(t)}$ transits to the next state, users' best response dynamics is similar to that of Protocol 2, i.e., converges to the Nash equilibrium corresponding to $C^{(t)}$ (which is fixed during the state holding time $T_h$). The detailed argument is given in Appendix~\ref{ap:th3} for completeness.
\end{proof}
Theorem \ref{th:wARM2} indicates that user's best response dynamics will converge and reach the time-evolving Nash equilibrium in the game corresponding to Protocol 3.


\section{Numerical Examples of the Edge Formation Game}
\label{sec: numerical}
In this section, we provide numerical results to illustrate how incorporating ARM (Algorithm 1) into the edge formation protocol (Protocol 2) reduces segregation, i.e., users form both intra- and inter-community edges at the Nash equilibrium of the corresponding game.
We also illustrate that higher \emph{acceptance probability} $C$ moves the corresponding game's Nash equilibrium closer towards social integration, i.e., users form more inter-community edges at the Nash equilibrium. 

\subsection{Convergence of Best Response Strategies Corresponding to Protocol 2}
\label{subsec:numerical1}
Recall that in Protocol 2, we propose that users alternatively play their best response strategies. In the following numerical example, we illustrate how
users' best response dynamics converges to
the Nash equilibrium of the corresponding edge formation game. We consider two communities each with $N=20$ users, and they play the edge formation game according to Protocol 2 for $T=20$ time steps. The ARM's \emph{acceptance probability} $C$ (defined in Algorithm 1) is set to be $0.8$. Fig. \ref{fig:best response converge} displays the convergence of two users' best responses to the unique Nash equilibrium $(0.75, 0.75)$. 

The steady state indicates that with ARM incorporated in the protocol and \emph{acceptance probability} $C>\frac{1}{2}$, users are incentivized to connect with users in different community, which illustrates the usefulness
of ARM in mitigating segregation.

\begin{figure}
	\centering
	\includegraphics[width=0.46\textwidth]{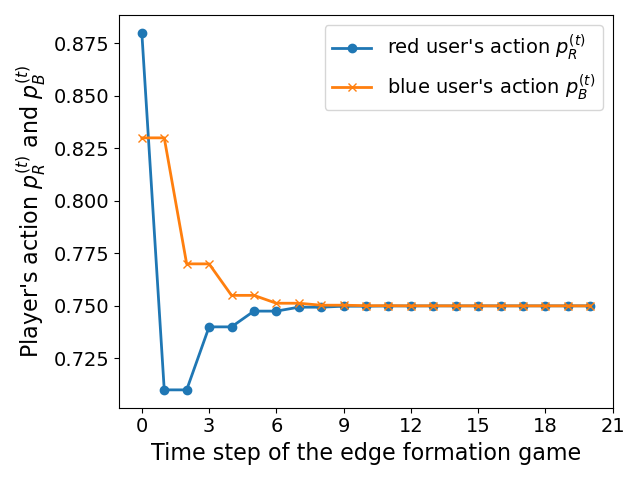}
	\caption{This figure shows the convergence rate of best responses of red player and blue player to the unique Nash equilibrium. With ARM (defined in Algorithm $1$) implemented and the acceptance probability set to $C=0.8$, users from both communities are incentivized to have $\frac{1}{4}$ of their total retweets to users from a different community at the equilibrium. The figure indicates that the best response strategy of the players converges to the Nash equilibrium in less than 10 iterations.} 
	\label{fig:best response converge}
\end{figure}

\subsection{The Effect of Varying Recommendation Acceptance Probability}
\label{subsec:numerical2}

Recall that we set ARM's \emph{acceptance probabilities} to be a fixed value in Protocol 2. In order to visualize the effect of different \emph{acceptance probabilities} (defined in Algorithm $1$) on the edge formation game, we vary the \emph{acceptance probability} from $\frac{2}{3}$ to $1$ and compare the resulting network structure 
sampled from the model at the Nash equilibrium. Fig. \ref{fig:graph structure} shows the social network at the Nash equilibrium
under different \emph{acceptance probabilities}, indicating that inter-community edges become denser with higher \emph{acceptance probability}.


\begin{figure}
	\centering
	\includegraphics[width=0.46\textwidth]{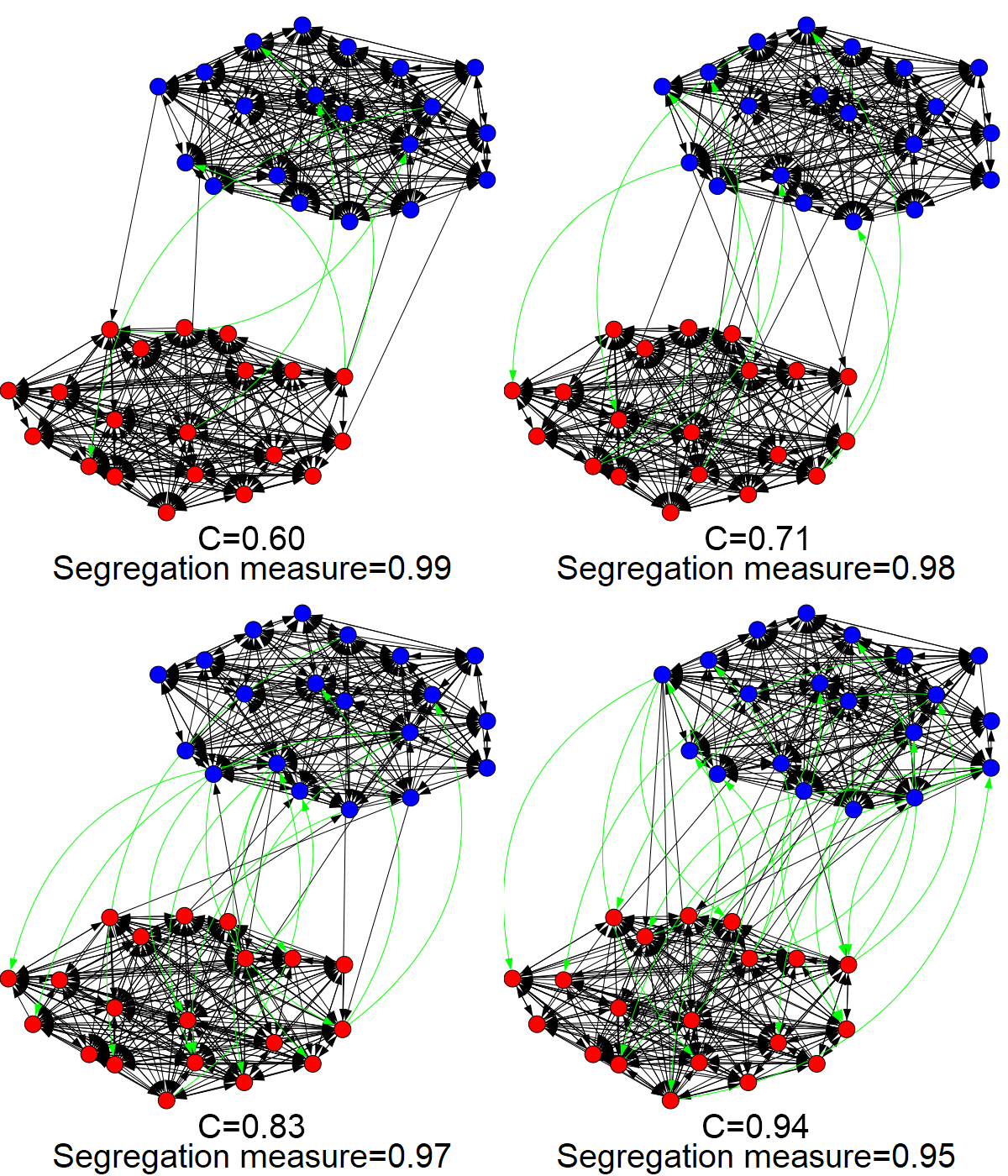}
	\caption{This figure shows the network structure of two communities after the edge formation game reaches Nash equilibrium with different acceptance probability $C$ (defined in Algorithm 1). The green curved inter-community edges represent the friend recommendations provided by ARM (Algorithm $1$). Viewing the four sub-figures in the order of left-to-right then top-to-bottom, the denser inter-community edges and lower modularity (discussed in Sec. \ref{subsec:numerical2} and Fig. \ref{fig:modularity}) indicate that users are more willing to involve in inter-community connections with higher $C$.}
	\label{fig:graph structure}
\end{figure}

We use the segregation measure defined in (\ref{eq:segregation index}) as a quantitative measurement of segregation. Networks with high segregation measure have dense intra-community connections but sparse inter-community connections. Fig. \ref{fig:modularity} shows that segregation measure decreases with higher \emph{acceptance probability} values, which agrees with our observation of the network structure. 

The results indicate that, during polarizing events when users tend to be more segregated, the network administrator can mitigate segregation by spending more efforts on friend recommendations (i.e., to achieve a higher \emph{acceptance probability} $C$) so that more inter-community edges can be formed. The results prove ARM's capability under different social settings. 


\begin{figure}
	\centering
	\includegraphics[width=0.46\textwidth]{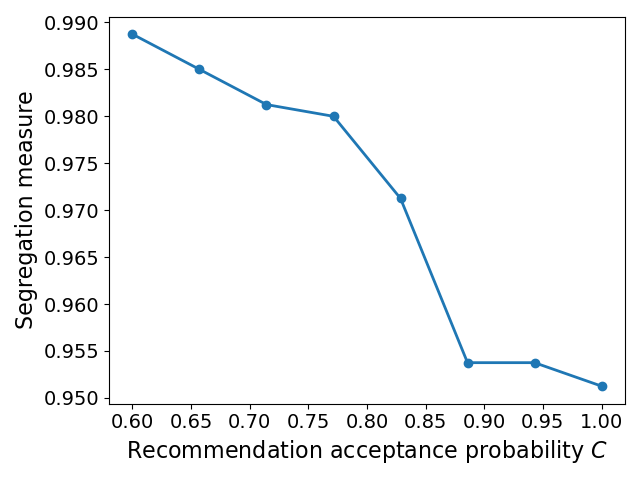}
	\caption{The figure shows that the segregation measure of the network (defined in (\ref{eq:segregation index})) drops with the increase of the \emph{acceptance probability} $C$ (defined in Algorithm 1). Higher segregation measure represents denser intra-community connections and sparser inter-community connections. The figure illustrates that higher \emph{acceptance probability} will better mitigate segregation in social networks.}
	\label{fig:modularity}
\end{figure}

\subsection{Validation of ARM in Opinion Dynamics Model} \label{subsec: opinion dynamics model}
To verify the efficacy of ARM in mitigating segregation, we implement it on the opinion dynamics model proposed in \cite{banisch2019opinion}. In the model, there are two opinions that agents can adopt and express. More precisely, an agent $i$ can adopt $o_i \in \{1, -1\}$. The model also accounts for how confident $i$ is about the two opinions by two real-valued terms $Q_i(1)$ and $Q_i(-1)$. At each step, an agent (say $i$) is chosen at random and expresses its opinion $o_i$ to a randomly chosen neighbor $j$, and $j$ responds to $i$'s expression with agreement or disagreement depending on $o_j$. The $Q_i(o)$ represents an internal evaluation of the opinions based on the social response $i$ obtains on expressing them. The value is updated as
\begin{equation}
    Q_i(o) = 
    \begin{cases}
    (1-\alpha) Q_i(o) + \alpha r_i & \textrm{if $o=$ expression}\\
    Q_i(o) & \textrm{else}
    \end{cases}
\end{equation}
where
\begin{equation}
    r_i = o_i o_j
\end{equation}
leading to a positive feedback for $o_i=o_j$ and to a negative one if $o_i \neq o_j$. The parameter $\alpha$ represents learning rate.

We incorporate ARM into this model as follows. When agent $i$'s expression is disagreed by $j$, $i$ will be recommended to connect with $j$'s neighbors who share the same opinion with $i$. Similar to ARM, the recommendation will be accepted with an acceptance probability $C$. Therefore although $i$ experiences disagreement with $j$, it will gain some confidence in its expression through agreement with $j$'s neighbors. The $Q_i(o)$ is updated as
\begin{equation}
    Q_i(o) = 
    \begin{cases}
    (1-\alpha) Q_i(o) + \alpha r_i & \textrm{if $o=$ expression}\\
    Q_i(o) & \textrm{else}
    \end{cases}
\end{equation}
where 
\begin{equation}
    r_i =
    \begin{cases}
    o_i o_j & \textrm{if $o_i = o_j$}\\
    o_i o_j + C \sum_{k\in N(j), k\neq i, o_k=o_i} o_i o_k &\textrm{if $o_i \neq o_j$} 
    \end{cases}
\end{equation}

We use the parameters provided by the authors\footnote{See Section 3.2 \cite{banisch2019opinion} for implementation details.}. 
The network is a random geometric graph with neighborhood radius $r=0.175$. 
There are $N=100$ agents. 
$Q$ take random initial values in $\text{Unif} (-0.5, 0.5)$. 
The learning rate $\alpha=0.05$. 
Follow the authors' implementation, we also set an exploration rate $\epsilon =0.1$, which measures the probability that agents express their less favorable opinions. 
We set the recommendation acceptance probability $C=0.9$.

\begin{figure}
	\centering
	\includegraphics[width=0.5\textwidth]{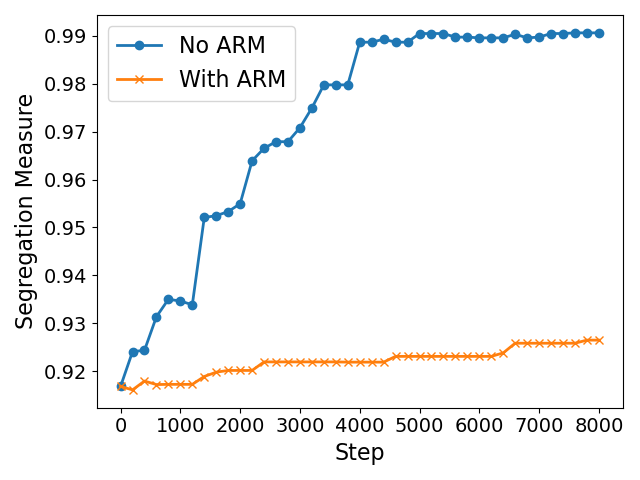}
	\caption{The figure shows the trend of the segregation measure (defined on a network in (\ref{eq:segregation index})) when agents follow the opinion dynamics model without and with ARM (defined in Section \ref{subsec: opinion dynamics model}). When agents follow the model without ARM, the segregation measure of the network keeps increasing and maintains at a high level in the final stage, whereas the segregation measure remains at a lower level when agents follow the model with ARM.}
	\label{fig:opinion model segregation}
\end{figure}

As shown in Fig. \ref{fig:opinion model segregation}, when agents follow the opinion dynamics model without ARM, the segregation measure of the network keeps increasing and maintains at a high level in the final stage, whereas the segregation measure remains at a relatively low level when agents follow the model is incorporated with ARM. Fig. \ref{fig:opinion model structure} compares the evolution of the network's community structure. The network is segregated into nearly disconnected components when ARM is not incorporated. On the contrary, the network is less segregated when ARM is incorporated. This numerical validation in opinion dynamics model verifies the efficacy of ARM in mitigating network segregation.
\begin{figure}
	\centering
	\includegraphics[width=0.5\textwidth]{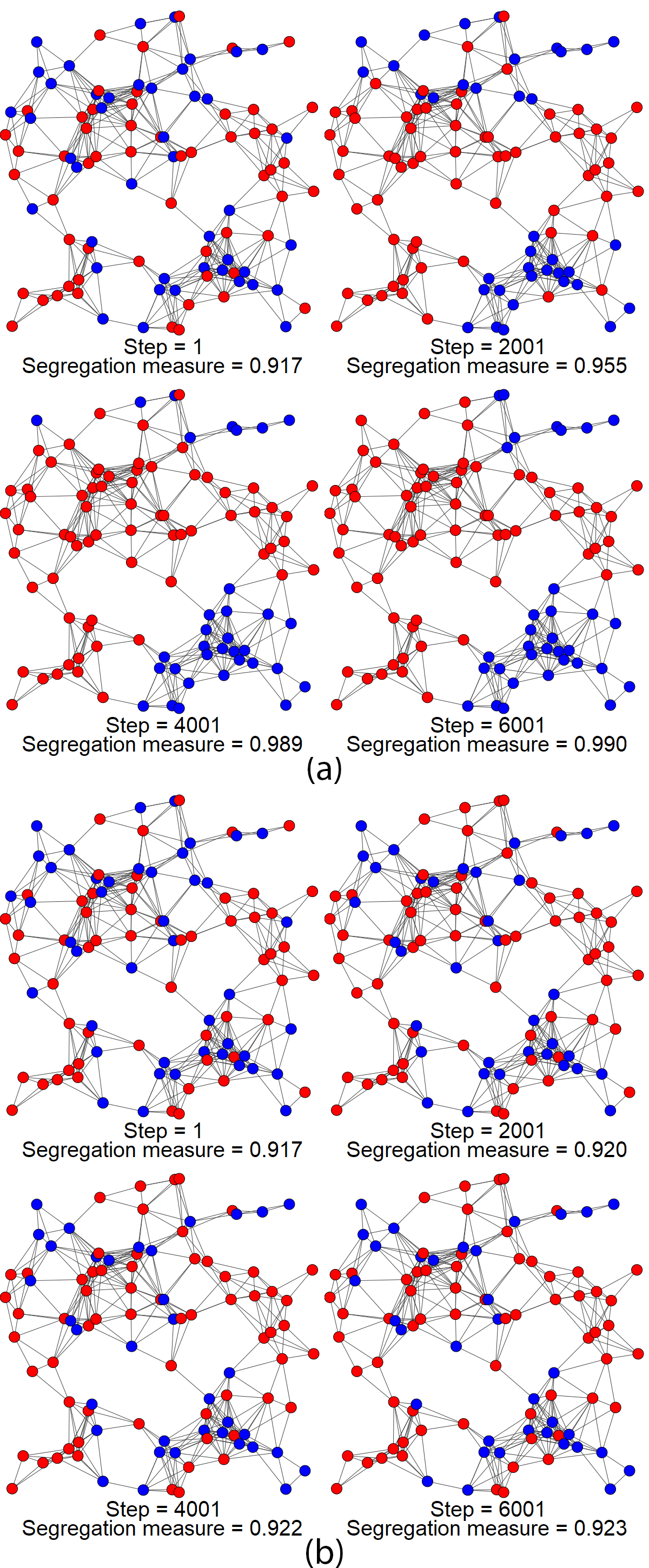}
	\caption{The figure shows the different evolution pattern of community structure in the model without and with ARM (defined in Section \ref{subsec: opinion dynamics model}). (a) shows that in the model without ARM, the community structure gradually segregates into nearly disconnected components; (b) shows that in the model with ARM, the community structure is less segregated.}
	\label{fig:opinion model structure}
\end{figure}

\section{Conclusions and Extensions}
\noindent
{\bf Conclusions: }
This paper considered an edge formation protocol on social networks represented by a \emph{directed stochastic block model} (DiSBM) to describe how users choose to connect with each other. 
The edge formation protocol represents how each individual chooses the connection probabilities in order to maximize a utility function that represents the tradeoff between homophily~(preference to be connected with one's own group) and popularity in the different community. Analysis of the game that corresponds to the best response based protocol shows that segregation is the unique Nash equilibrium. 
We then proposed an algorithmic recommendation mechanism (ARM) to mitigate segregation. ARM recommends users from different communities to form weak ties and provides incentives to those that form them. Assuming each recommendation suggested by the ARM is accepted with an \emph{acceptance probability}, we show that the segregation level at the Nash equilibrium of the corresponding game depends on it. Thus, the segregation of the network can be controlled by introducing the ARM. We further extend our results to the case where the \emph{acceptance probability} itself has Markovian dynamics and illustrate how the individuals in the network reach the time-evolving Nash equilibrium of the resulting game. Thus, our results provide a novel mechanism design perspective into the problem of mitigating segregation in social networks.

\vspace{0.2cm}
\noindent
{\bf Limitations and Extensions: }
The proposed edge formation game and its analysis can be extended to further contexts in several directions.
\begin{enumerate}[\hspace{0in}1)]
    \item The edge formation protocol proposed in this paper assumes homogeneity among users in one community, i.e., users' community information solely determines their edge formation probabilities. 
    An interesting future direction is to incorporate heterogeneity (e.g., preferential attachment with fitness) into the network model and apply methods such as friendship paradox sampling \cite{nettasinghe2019friendship,nettasinghe2019maximum,krishnamurthy2019information} to assign different weights to the 2 hop connections between different pairs of users. 
    
    

    \item Another interesting direction is to consider ARM's parameter (e.g., the recommendation acceptance probability) depends on users' actions and incorporates a feedback law, i.e., ARM raises the acceptance probability when users become more segregated in network. Such extensions might be helpful in analyzing the existence of Markov Perfect Equilibrium (MPE) of the corresponding stochastic game, which enhances the understanding of users' long-term strategy in an evolving social network.
    
    \item The segregation measure (\ref{eq:segregation index}) based on inter-community edges can be further applied in other network topics such as community detection\cite{li2012identifying,li2019dynamical,li2020optimization}, and estimation of latent network factors\cite{li2020optimal}.

    
\end{enumerate}

\appendices
\section{Proofs of Theorems}
\label{sec:appendix}

\subsection{Proof of Theorem \ref{th:woARM}}
\label{ap:th1}
In Protocol 1, the utility function (\ref{eq:utility_function}) of each user in the same community is a random variable with the probability distribution induced by the DiSBM model, which has the same expected value. Users in the same community can be viewed as independent copies of a single player. 
Thus the edge formation game can be reduced to a two-player game, where each player represents users in one community.
The utility of the two players, $U_R^{(t)}$ and $U_B^{(t)}$, are respectively the expected value of the utility function (\ref{eq:utility_function}) for red and blue users:
\begin{equation}
\label{eq:expected utility 1}
\begin{split}
    U_R^{(t)} 
    &= p_{R}^{(t)} - p_{B}^{(t)} 
\end{split}
\end{equation}
\begin{equation}
\label{eq:expected utility 2}
\begin{split}
    U_B^{(t)} 
    &= p_{B}^{(t)} - p_{R}^{(t)} 
\end{split}
\end{equation}

It is shown from (\ref{eq:expected utility 1}) that $p_R^{(t)} = 1$ is red player's strictly dominant strategy. 
This is equivalent to say, at red player's turn, it will choose $p_{R}^{(t)}=1$ as its best response strategy regardless of blue player's previous action, and will never switch again. Similar result holds for blue player. 
Therefore the game converges in 2 time steps, i.e., $p_R^{(t)} = p_B^{(t)} = 1$ for all $t = 2,3,\cdots$. 

Thus, the game corresponding to Protocol 1 converges to its unique Nash equilibrium $(p_{R}^{(t)}, p_{B}^{(t)}) = (1, 1)$ after each player acted once, i.e., users only form intra-community edges and the social network is segregated into echo chambers after $t=2$.

\subsection{Proof of Theorem \ref{th:wARM}}
\label{ap:th2}
Similar to the proof of Theorem \ref{th:woARM}, we assume users in one community are independent copies of a single player and
reduce the game with ARM to a two-player game. Based on (\ref{eq:utility_function_2}), the utility of the red player and blue player in the two-player game are respectively:
\begin{equation}
\label{eq:expected utility 3}
\resizebox{.89\hsize}{!}{$
\begin{aligned}
    U_R^{(t)} 
    &= p_{R}^{(t)} - p_{B}^{(t)} + C\big[p_{B}^{(t)}(2 - p_{R}^{(t)} - p_{B}^{(t)}) + p_{R}^{(t)}(1 - p_{R}^{(t)})\big]
\end{aligned}
$}
\end{equation}
\begin{equation}
\label{eq:expected utility 4}
\resizebox{.89\hsize}{!}{$
\begin{aligned}
    U_B^{(t)} 
    &= p_{B}^{(t)} - p_{R}^{(t)} + C\big[p_{R}^{(t)}(2 - p_{R}^{(t)} - p_{B}^{(t)}) + p_{B}^{(t)}(1 - p_{B}^{(t)})\big]
\end{aligned}
$}
\end{equation}
A detailed derivation can be found in Appendix \ref{subsec:appendix2}.

The second order mixed derivatives of (\ref{eq:expected utility 3}, \ref{eq:expected utility 4}) yields 
\begin{equation}
\pdv{U_R^{(t)}}{p_R^{(t)}}{p_B^{(t)}} = \pdv{U_B^{(t)}}{p_R^{(t)}}{p_B^{(t)}} = -C
\end{equation}
Since $C>0$, the game is submodular, which 
guarantees that players' best response strategies converge to a Nash equilibrium \cite{amir2005supermodularity}. This submodular game also has quadratic utility, which makes it similar to the Bertrand game \cite{milgrom1990rationalizability} and thus can be analyzed in a similar approach. 

Note that
\begin{equation}
    \pdv{U_R^{(t)}}{p_R^{(t)}} = 1-C(2p_R^{(t)}+p_B^{(t)}-1)
\end{equation}
\begin{equation}
    \pdv{U_B^{(t)}}{p_B^{(t)}} = 1-C(2p_B^{(t)}+p_R^{(t)}-1)
\end{equation}

We apply iterated strict dominance to attain the Nash equilibrium. Let red player's initial best response set $b^{(0)}_{R}=(0, 1]$.
\begin{enumerate}[\hspace{0in}1)]
    \item If $p_R^{(t)} < \frac{1}{2C}$, then $\pdv{U_R^{(t)}}{p_R^{(t)}} \geq 0$ $\rightarrow$ any $p_R^{(t)} < \frac{1}{2C}$ is strictly dominated.
    \item If $p_R^{(t)} > \frac{1}{2C}+\frac{1}{2}$, then $\pdv{U_R^{(t)}}{p_R^{(t)}} \leq 0$ $\rightarrow$ any $p_R^{(t)} > \frac{1}{2C}+\frac{1}{2}$ is strictly dominated.
\end{enumerate}

The following analysis depends on two different cases of the acceptance probability $C$.
\begin{enumerate}[\hspace{0in}1)]
    \item If $0 < C\leq \frac{1}{2}$, then $\frac{1}{2C}\geq 1$, i.e., $\pdv{U_R^{(t)}}{p_R^{(t)}} \geq 0$ holds for red player's possible actions $p_R^{(t)} \in (0, 1]$. Thus red player's best response is
    \begin{equation}
        b_R(p_R^{(t)}) = 1
    \end{equation}
    Similarly, blue player's best response is
    \begin{equation}
        \textrm{b}_{B}(p_{R}^{(t)}) = 1
    \end{equation}
    Therefore, the Nash equilibrium of the game is 
    \begin{equation}
        (p_{R}^{(t)}, p_{B}^{(t)}) = (1, 1)
    \end{equation}
    i.e., the social network is in segregation.
    
    \item If $\frac{1}{2} < C\leq 1$, then $\frac{1}{2C}+\frac{1}{2}\geq 1$. Thus, after one iteration, red player's remaining undominated strategy set, i.e, the best response set is $b^{1}_{R}=[\frac{1}{2C}, 1]$.
    
    Let the set after $i$ iterations be $b^{i}_{R}=[\underbar{b}^{i}, \bar{b}^{i}]$, where
    \begin{equation}
        \underbar{b}^{i} = \frac{1}{2}(\frac{1}{C}+1-\bar{b}^{i-1})
    \end{equation}
    \begin{equation}
        \bar{b}^{i} = \frac{1}{2}(\frac{1}{C}+1-\underbar{b}^{i-1})
    \end{equation}
    We can show that
    \begin{equation}
        \displaystyle\lim_{i\to \infty} \underbar{b}^{i} = \bar{b}^{i} = \frac{1}{3C} + \frac{1}{3}
    \end{equation}
    
    Therefore, by iteratively applying the best responses and eliminating strictly dominated strategies, the best response strategy converges to the unique Nash equilibrium of the game, which is 
    \begin{equation}
    \label{eq:NE_C}
        (p_R^{(t)}, p_B^{(t)}) = (\frac{1}{3C} + \frac{1}{3}, \frac{1}{3C} + \frac{1}{3})
    \end{equation}
     The Nash equilibrium ($p_{R}^{(t)}<1,\,p_{B}^{(t)}<1$) indicates that ARM incentivizes both red and blue users to form inter-community edges, leading to social integration.
\end{enumerate}

\subsection{Proof of Theorem \ref{th:wARM2}}
\label{ap:th3}
Following the proof of Theorem \ref{th:woARM} and \ref{th:wARM}, we assume users in one community are independent copies of a single player and
reduce the game to a two-player game. We define the two-player stochastic game resulting from Protocol 3 as follows:
\begin{enumerate}[\hspace{0in}1)]
    \item A state of the game, $C^{(t)}$, representing the \emph{acceptance probability} at time $t$. 
    \item Actions of red player and blue player $p_{R}^{(t)}$, $p_{B}^{(t)}\in (0, 1]$, representing their best response strategies at time $t$.
    \item One-stage payoffs for red player and blue player  which are the same as (\ref{eq:expected utility 3}) and (\ref{eq:expected utility 4}), except that the fixed $C$ is substituted by the time-variant $C^{(t)}$
    \begin{equation}
    \resizebox{.8\hsize}{!}{$
    \begin{aligned}
    U_R^{(t)} 
     = & p_{R}^{(t)} - p_{B}^{(t)} \\
        & + C^{(t)}\big[p_{B}^{(t)}(2 - p_{R}^{(t)} - p_{B}^{(t)}) + p_{R}^{(t)}(1 - p_{R}^{(t)})\big]
    \end{aligned}
    $}
    \end{equation}
    \begin{equation}
    \resizebox{.89\hsize}{!}{$
    \begin{aligned}
    U_B^{(t)} 
    = & p_{B}^{(t)} - p_{R}^{(t)}\\
        & + C^{(t)}\big[p_{R}^{(t)}(2 - p_{B}^{(t)} - p_{R}^{(t)}) + p_{B}^{(t)}(1 - p_{B}^{(t)})\big]
    \end{aligned}
    $}
    \end{equation}
    \item A transition rule of $C^{(t)}$ specified in (\ref{eq: evolve C}), which is conditionally independent of players' actions.
    
    \item A discount factor $\gamma \in (0, 1)$.
    \item Value functions of red player and blue player representing their discounted sum of payoff from time $t$ 
    \begin{equation}
        V_{R}^{(t)} = \sum_{\tau=t}^{\infty} \gamma^{\tau-t} U_{R}^{(\tau)}
    \end{equation}
    \begin{equation}
        V_{B}^{(t)} = \sum_{\tau=t}^{\infty} \gamma^{\tau-t} U_{B}^{(\tau)}
    \end{equation}
\end{enumerate}

Based on the above definition of the stochastic game, we aim to derive the player's optimal strategy under the condition that the \emph{acceptance probability} evolves as a semi-Markov process. Recall that in Protocol 2 where the \emph{acceptance probability} is fixed, (in the proof of Theorem \ref{th:wARM}) we have illustrated the convergence of players' best response dynamics to the steady state corresponding to the game's Nash equilibrium. Considering that the state holding time $T_h$ (defined in (\ref{eq: evolve C}))
is large compared with the time that players take action, we claim that in Protocol 3, players' best response dynamics converges to the Nash equilibrium corresponding to $C^{(t)}$ before it transits.

With this claim, what remains to be proved is how players adapt their actions when $C^{(t)}$ transits. Below we illustrate players' strategy at $t$ when $C^{(t)}$ transits to another state $C^{(t+1)}$.

Bellman's dynamic programming recursion yields \cite{krishnamurthy2016partially}:
\begin{equation}
\label{eq: Bellman}
\begin{split}
    V_{R}^{(t)} & = \max_{p_{R}^{(t)}} U_{R}^{(t)}   
    + \gamma \sum_{C^{(t+1)}} P(C^{(t+1)}|C^{(t)}) V_{R}^{(t+1)}  
\end{split}
\end{equation}
In (\ref{eq: Bellman}), $V_{R}^{(t)}$ and $V_{R}^{(t+1)}$ denotes the red user's value function at time $t$ and $t+1$ respectively, assuming that it took the best response $p_R^{(t)}$ at time $t$, which corresponds to the Nash equilibrium of the game with $C^{(t)}$.

The transition dynamics of $C^{(t)}$ are conditionally independent of players' actions. Only the first term in (\ref{eq: Bellman}), i.e., the one-stage payoff $U_{R}^{(t)}$, depends on players' actions. Thus red player's optimal strategy is to myopically maximize the one-stage payoff without concerning about the transition of $C^{(t)}$. Similar result holds for blue player's optimal strategy. Based on the Nash equilibrium (\ref{eq:NE_C}) derived in Sec. \ref{subsec: game analysis 2}, $(\frac{1}{3C^{(t)}}+\frac{1}{3}, \frac{1}{3C^{(t)}}+\frac{1}{3})$ are the optimal strategies for red and blue users. Thus, players' best response dynamics will converge to $(\frac{1}{3C^{(t)}}+\frac{1}{3}, \frac{1}{3C^{(t)}}+\frac{1}{3})$, i.e., they will reach the time-evolving Nash equilibrium in the resulting game.

\section{Derivation of User's Expected Utility with ARM in Equation (\ref{eq:expected utility 3}, \ref{eq:expected utility 4})}
\label{subsec:appendix2}
Take red user as an example. The expected number of inter-community edges formed by ARM for any red user is 
\begin{equation}
\label{eq: appendix 1}
\resizebox{.89\hsize}{!}{$
\begin{aligned}
    & \mathbb{E}\left\{ \displaystyle \sum_{j\in B} \Big\{\scriptstyle C \big[\scriptstyle 1-D^{(t)}(i, j) \big] 
    \frac{\displaystyle \sum_{j'\in B} \scriptstyle \big[D^{(t)}(i, j') + D^{(t)}(j', i)\big] 
    \scriptstyle S^{(t)}(j', j)}{\scriptstyle N-1} \Big\} \right\} \\
    & = C(1-\frac{(1-p_{B}^{(t)})}{N}) N \big(\frac{1-p_{R}^{(t)}}{N} + \frac{1-p_{B}^{(t)}}{N} \big) p_{B}^{(t)} \\ 
    & = CN\big(\frac{1-p_{R}^{(t)}}{N} + \frac{1-p_{B}^{(t)}}{N} \big) p_{B}^{(t)}  \\
    & = C(2-p_{R}^{(t)}-p_{B}^{(t)})p_{B}^{(t)}
    \end{aligned}
$}
\end{equation}
where $\mathbb{E}$ denotes expectation over the probability distribution induced by the DiSBM.

On the second line of (\ref{eq: appendix 1}), $C$ is the \emph{acceptance probability}, $N$ is the number of blue users, $(1-\frac{1-p_{B}^{(t)}}{N})$ is the probability that the recommendation target has not followed $i$ yet, $\big(\frac{1-p_{R}^{(t)}}{N} + \frac{1-p_{B}^{(t)}}{N} \big) p_{B}^{(t)}$ is the probability that $i$ is recommended by ARM to any blue user. The third line satisfies when $N$ is large, i.e., $\frac{1-p_{B}^{(t)}}{N} \approx 0$. 

Similarly, $i$'s expected reward for connecting users from two communities (e.g. a red user connects its red followers to its blue friends) during ARM recommendations is 
\begin{equation}
\label{eq: appendix 2}
    \begin{split}
    & \mathbb{E}\left\{\displaystyle \sum_{j\in B} \displaystyle \sum_{i'\in R} \Big\{\scriptstyle C \big[\scriptstyle 1-D^{(t)}(j, i') \big] \scriptstyle \big[\frac{S^{(t)}(i, i') D^{(t)}(j, i)}{\scriptstyle N-1} \big] \Big\}\right\}  \\
    & = C(1-\frac{1-p_{R}^{(t)}}{N})  N\frac{p_R^{(t)}\frac{1-p_{R}^{(t)}}{N}(N-1)}{N-1} \\
    & = C(1-p_{R}^{(t)})p_{R}^{(t)}
    \end{split}
\end{equation}

Combining (\ref{eq: appendix 1}, \ref{eq: appendix 2}) for the additional reward, and (\ref{eq:expected utility 1}, \ref{eq:expected utility 2}) for the expected utility of users in the game without ARM, we can derive the expected utility of users, i.e., the utility of red and blue player, in (\ref{eq:expected utility 3}, \ref{eq:expected utility 4}). 

\ifCLASSOPTIONcompsoc
  \section*{Acknowledgments}
\else
  \section*{Acknowledgment}
\fi

This research was supported in part by the  U. S. Army Research Office under grants W911NF-21-1-0093 and W911NF-19-1-0365.

\ifCLASSOPTIONcaptionsoff
  \newpage
\fi



%
\bstctlcite{IEEEexample:BSTcontrol}
\bibliographystyle{IEEEtran}
\end{document}